\def\version{J. Math. Phys. 50, 053509 (2009)}

\documentclass[reqno]{amsart}

\usepackage{amssymb,epsf}

\def\be{\begin{equation}}
\def\ee{\end{equation}}
\def\ba{\begin{align}}
\def\ea{\end{align}}
\def\bsplit{\begin{split}}
\def\esplit{\end{split}}
\def\bm{\begin{multline}}
\def\eem{\end{mutline}}
\def\bfig{\begin{figure}[htb]}
\def\efig{\end{figure}}

\numberwithin{equation}{section}
\newtheorem{theorem}{Theorem}[section]
\newtheorem{proposition}[theorem]{Proposition}
\newtheorem{lemma}[theorem]{Lemma}

\newtheorem{assumption}{Assumption}

\newcommand{\nn}{\nonumber}

\renewcommand{\leq}{\;\leqslant\;}
\renewcommand{\geq}{\;\geqslant\;}
\newcommand{\dd}{{\rm d}}
\newcommand{\e}[1]{\,{\rm e}^{#1}\,}

\newcommand{\sumtwo}[2]{\sum_{\substack{#1 \\ #2}}}

\def\Tr{{\operatorname{Tr\,}}}
\def\dist{{\operatorname{dist\,}}}
\def\Re{{\operatorname{Re\,}}}

\newcommand{\upchi}{\raise 2pt \hbox{$\chi$}}

\newcommand{\caA}{{\mathcal A}}\newcommand{\caC}{{\mathcal C}}\newcommand{\caF}{{\mathcal F}}\newcommand{\caG}{{\mathcal G}}\newcommand{\caP}{{\mathcal P}}\newcommand{\caS}{{\mathcal S}}\newcommand{\caT}{{\mathcal T}}\newcommand{\caX}{{\mathcal X}}

\newcommand{\bbB}{{\mathbb B}}\newcommand{\bbR}{{\mathbb R}}\newcommand{\bbX}{{\mathbb X}}\newcommand{\bbZ}{{\mathbb Z}}

\begin{document}

{\hfill\small \version} \vspace{2mm}

\title[Abstract cluster expansion with applications]{Abstract cluster expansion with applications to statistical mechanical systems}

\author{Suren Poghosyan}
\address{Suren Poghosyan \hfill\newline
\indent Institute of Mathematics \hfill\newline
\indent Armenian National Academy of Science \hfill\newline
\indent Marshal Bagramian 24-B \hfill\newline
\indent Yerevan, 375019, Armenia
}
\email{suren.poghosyan@unicam.it}

\author{Daniel Ueltschi}
\address{Daniel Ueltschi \hfill\newline
\indent Department of Mathematics \hfill\newline
\indent University of Warwick \hfill\newline
\indent Coventry, CV4 7AL, England \hfill\newline
{\small\rm\indent http://www.ueltschi.org}
}
\email{daniel@ueltschi.org}

\begin{abstract}
We formulate a general setting for the cluster expansion method and we discuss sufficient criteria for its convergence. We apply the results to systems of classical and quantum particles with stable interactions.

\vspace{1mm}

\vspace{1mm}
\noindent
{\it 2000 Math.\ Subj.\ Class.:} 82B05, 82B10, 82B20, 82B21, 82B26\\
\end{abstract}

\maketitle

\section{Introduction}

The method of cluster expansions was introduced in the 1930's in statistical mechanics in order to study gases of classical interacting particles. Its main achievement, from the point of view of physics, may be the derivation of the van der Waals equation of state for weakly interacting systems. The method was made rigorous by mathematical-physicists in the 1960's, see \cite{Rue} and references therein.

The method split afterwards. One branch involves continuous systems, with applications to classical systems \cite{Pen,MP,BrF}, quantum systems \cite{Gin1,Gin,PZ}, or quantum field theory \cite{GJS,Mal,BaF,BK}. The other branch involves polymer systems, i.e.\ discrete systems with additional internal structure \cite{GK,Dob,BZ,Mir,Sok,FP,JPS}. An important step forward was the article of Koteck\'y and Preiss with its simplified setting and its elegant condition for the convergence of the cluster expansion \cite{KP}.

The methods for proving the convergence are diverse. Let us mention the study of Kirkwood-Salsburg equations that involves correlation functions, see \cite{Rue} and references therein; the algebraic approach of Ruelle \cite{Rue}; combinatorial approaches using tree identities \cite{Pen,BrF,BaF,BK}; inductions for discrete systems \cite{Dob,BZ,Mir}.

Important and useful surveys were written by Brydges \cite{Bry}, Pfister \cite{Pfi}, Abdesselam and Rivasseau \cite{AR}. Recent articles have been devoted to combinatorial aspects \cite{Sok,Far,JPS} and to weakening the assumptions \cite{FP,Far,Pro}.

The method of cluster expansions applies when the objects do not interact much; this the case when they are far apart (low density), or when interactions are weak. An extension of the criterion of \cite{KP} that takes into account these two aspects was proposed in \cite{Uel}; it applies to both discrete and continuous systems.

All abstract (i.e.\ general) approaches involve restrictions that correspond to repulsive interactions. Yet the old results for classical and quantum systems only assume {\it stable} interactions, that may include an attractive part. The aim of the present article is to propose a general approach that applies to discrete and continuous systems with repulsive or stable interactions. Our proof is split into several independent steps and this helps clarify the situation.

The setting and the results are presented in Section \ref{sec clexp}. We consider applications to classical systems of particles in Section \ref{sec class gas}, to polymer systems in Section \ref{sec class polymers}, and to the quantum gas in Section \ref{sec quantum gas}. A fundamental tree estimate is derived in Section \ref{sec tree estimate}, and the theorems of Section \ref{sec clexp} are proved in Section \ref{sec proofs}.

\section{Cluster expansions}
\label{sec clexp}

We consider a set $\bbX$ whose elements may represent widely different objects --- in the three applications considered in this article, an element $x\in\bbX$ represents (i) the position of a classical particle, (ii) a polymer, i.e.\ a connected set of $\bbZ^d$, and (iii) a closed Brownian bridge. For the general abstract theory, we assume the structure of a measure space, $(\bbX,\caX,\mu)$, with $\mu$ a complex measure. We denote $|\mu|$ the total variation (absolute value) of $\mu$. Let $u$ and $\zeta$ be complex measurable symmetric functions on $\bbX \times \bbX$, that are related by the equation
\be
\zeta(x,y) = \e{-u(x,y)}-1.
\ee
We allow the real part of $u$ to take the value $+\infty$, in which case $\zeta(x,y)=-1$. In typical applications $u(x,y)$ represents the interactions between $x$ and $y$, and the value $+\infty$ corresponds to a hard-core repulsion. We define the ``partition function" by
\be
\label{def Z}
Z = \sum_{n\geq0} \frac1{n!} \int\dd\mu(x_1) \dots \int\dd\mu(x_n) \exp\Bigl\{ -\sum_{1\leq i<j\leq n} u(x_i,x_j) \Bigr\},
\ee
or, equivalently,
\be
Z = \sum_{n\geq0} \frac1{n!} \int\dd\mu(x_1) \dots \int\dd\mu(x_n) \prod_{1\leq i<j\leq n} \bigl( 1 + \zeta(x_i,x_j) \bigr).
\ee
The term $n=0$ of the sums is understood to be 1.

The main goal of cluster expansions is to express the partition function as the exponential of a convergent series of ``cluster terms". The main difficulty is to prove the convergence. We first assume that the potential $u$ is stable.

\begin{assumption}
\label{ass stability}
There exists a nonnegative function $b$ on $\bbX$ such that, for all $n$ and almost all $x_1,\dots,x_n \in \bbX$,
\[
\prod_{1\leq i<j \leq n} \bigl| 1 + \zeta(x_i,x_j) \bigr| \leq \prod_{i=1}^n \e{b(x_i)}.
\]
\end{assumption}

In  other words, we assume the lower bound
\be
\sum_{1\leq i<j \leq n} \Re u(x_i,x_j) \geq -\sum_{i=1}^n b(x_i).
\ee
When the function $b$ is constant, this is the usual definition of stability. ``Almost all" means that, for given $n$, the set of points where the condition fails has measure zero with respect to the product measure $\otimes^n \mu$. If $\bbX$ is countable, the condition must be satisfied for all $x_1,\dots,x_n$ such that $\mu(x_i)\neq0$.

The second condition deals with the strength of interactions.

\begin{assumption}
\label{ass KP crit}
There exists a nonnegative function $a$ on $\bbX$ such that for almost all $x \in \bbX$,
\[
\int\dd|\mu|(y) \, |\zeta(x,y)| \e{a(y)+2b(y)} \leq a(x).
\]
\end{assumption}

In order to guess the correct form of $a$, one should consider the left side of the equation above with $a(y) \equiv 0$. The integral may depend on $x$; a typical situation is that $x$ is characterized by a length $\ell(x)$, which is a positive number, so that the left side is roughly proportional to $\ell(x)$. This suggests to try $a(x) = c\ell(x)$, and one can then optimize on the value of $c$.

We also consider an alternate criterion that involves $u$ rather than $\zeta$. It is inspired by the recent work of Procacci \cite{Pro}. Let
\be
\overline u(x,y) = \begin{cases} u(x,y) & \text{if } \Re u(x,y) \neq \infty, \\ 1 & \text{if } \Re u(x,y) = \infty. \end{cases}
\ee

\renewcommand{\theassumption}{2'}
\begin{assumption}
\label{ass tree crit}
There exists a nonnegative function $a$ on $\bbX$ such that for almost all $x \in \bbX$,
\[
\int\dd|\mu|(y) \, |\overline u(x,y)| \e{a(y)+b(y)} \leq a(x).
\]
\end{assumption}

For positive $u$ we can take $b(x) \equiv 0$; and since $1-\e{-u} \leq u$, Assumption \ref{ass KP crit} is always better than Assumption \ref{ass tree crit}. We actually conjecture that, together with Assumption \ref{ass stability}, a sufficient condition is
\be
\int\dd|\mu|(y) \, \min\bigl( |\zeta(x,y)|,|u(x,y)| \bigr) \e{a(y)+b(y)} \leq a(x).
\ee
That is, it should be possible to combine the best of both assumptions. In this respect Assumption \ref{ass KP crit} is optimal in the case of positive potentials, and Assumption \ref{ass tree crit} is optimal in the case of hard core plus negative potentials.

We denote by $\caG_n$ the set of all graphs with $n$ vertices (unoriented, no loops) and $\caC_n \subset \caG_n$ the set of connected graphs with $n$ vertices. We introduce the following combinatorial function on finite sequences $(x_1,\dots,x_n)$ of elements of $\bbX$:
\be
\label{def comb fct}
\varphi(x_1,\dots,x_n) = \begin{cases} 1 & \text{if } n=1, \\ \frac1{n!} \sum_{G \in \caC_n} \prod_{\{i,j\} \in G} \zeta(x_i,x_j) & \text{if } n\geq2. \end{cases}
\ee
The product is over edges of $G$.

\begin{theorem}[Cluster expansions]
\label{thm clexp}
Suppose that Assumptions \ref{ass stability} and \ref{ass KP crit}, or \ref{ass stability} and \ref{ass tree crit}, hold true. We also suppose that $\int\dd|\mu|(y)| \e{a(y)+2b(y)} < \infty$. Then we have
\be
\label{clexp}
Z = \exp\Bigl\{ \sum_{n\geq1} \int\dd\mu(x_1) \dots \dd\mu(x_n) \, \varphi(x_1,\dots,x_n) \Bigr\}.
\ee
The term in the exponential converges absolutely. Furthermore, for almost all $x_1 \in \bbX$, we have the following estimate
\be
\label{the bound}
\sum_{n\geq2} n \int\dd|\mu|(x_2) \dots \int\dd|\mu|(x_n) \, |\varphi(x_1,\dots,x_n)| \leq (\e{a(x_1)}-1) \e{2b(x_1)}.
\ee
\end{theorem}

(Under Assumption \ref{ass tree crit}, Eq.\ \eqref{the bound} holds with $\e{b(x_1)}$ instead of $\e{2b(x_1)}$.)

Let us turn to correlation functions. We only consider one-point and two-point correlation functions since these are the most useful and expressions become more transparent. We refer to \cite{Uel} for more general functions. First, we define the {\it unnormalized one-point correlation function} by
\be
Z(x_1) = \sum_{n\geq1} \frac1{(n-1)!} \int\dd\mu(x_2) \dots \int\dd\mu(x_n) \prod_{1\leq i<j\leq n} \bigl( 1 + \zeta(x_i,x_j) \bigr)
\ee
(the term $n=1$ is 1 by definition). And we define the {\it unnormalized two-point correlation function} by
\be
Z(x_1,x_2) = \sum_{n\geq2} \frac1{(n-2)!} \int\dd\mu(x_3) \dots \int\dd\mu(x_n) \prod_{1\leq i<j\leq n} \bigl( 1 + \zeta(x_i,x_j) \bigr)
\ee
(the term $n=2$ is equal to $1+\zeta(x_1,x_2)$). Notice that $Z(x_1)$ can be viewed as a regular partition function, given by Eq.\ \eqref{def Z}, but with the modified measure $(1 + \zeta(x_1,x)) \mu(x)$ instead of $\mu(x)$. The {\it normalized} correlation functions are $Z(x_1)/Z$ and $Z(x_1,x_2)/Z$. As is shown in Theorem \ref{thm correlation functions}, they can be expressed using the ``cluster functions"
\be
\label{def Z1 hat}
\hat Z(x_1) = \sum_{n\geq 1} n \int \dd\mu(x_2) \dots \int\dd\mu(x_n) \, \varphi(x_1,\dots,x_n),
\ee
and
\be
\label{def Z2 hat}
\hat Z(x_1,x_2) = \sum_{n\geq 2} n(n-1) \int \dd\mu(x_3) \dots \int\dd\mu(x_n) \, \varphi(x_1,\dots,x_n).
\ee
Notice that $|\hat Z(x_1)| \leq \e{a(x_1)+2b(x_1)}$ by \eqref{the bound}.

\begin{theorem}[Correlation functions]
\label{thm correlation functions}
Under the same assumptions as in Theorem \ref{thm clexp}, we have
\[
\begin{split}
\frac{Z(x_1)}Z &= \hat Z(x_1), \\
\frac{Z(x_1,x_2)}Z &= \hat Z(x_1) \hat Z(x_2) + \hat Z(x_1,x_2).
\end{split}
\]
\end{theorem}

In statistical mechanics, the relevant expression is the {\it truncated two-point correlation function}
\[
\frac{Z(x_1,x_2)}Z - \frac{Z(x_1) \, Z(x_2)}{Z^2}.
\]
When the cluster expansion converges, it is equal to $\hat Z(x_1,x_2)$ by the theorem above. This function usually provides an order parameter for phase transitions and it is useful to estimate its decay properties.

\begin{theorem}[Decay of correlations]
\label{thm decay correlation functions}
If Assumptions \ref{ass stability} and \ref{ass KP crit} hold true, we have for almost all $x,y \in \bbX$,
\bm
|\hat Z(x,y)| \leq \e{a(y)+2b(y)} \Bigl[ |\zeta(x,y)| \e{a(x)+2b(x)} + \\
+ \sum_{m\geq1} \int\dd|\mu|(x_1) \dots \int\dd|\mu|(x_m) \prod_{i=0}^m |\zeta(x_i,x_{i+1})| \e{a(x_i)+2b(x_i)} \Bigr] \nn
\end{multline}
(with $x_0 \equiv x$ and $x_{m+1} \equiv y$). If Assumptions \ref{ass stability} and \ref{ass tree crit} hold true, we have the same bound but with $|\overline u(\cdot,\cdot)|$ instead of $|\zeta(\cdot,\cdot)|$, and $\e{b(\cdot)}$ instead of $\e{2b(\cdot)}$.
\end{theorem}

In many applications the functions $\zeta(x,y)$ and $u(x,y)$ depend on the difference $x-y$ (this assumes that $\bbX$ has additional structure, namely that of a group). The estimates for $|\hat Z(x,y)|$ are given by convolutions.

The theorems of this section are proved in Section \ref{sec proofs}.

\section{Classical gas}
\label{sec class gas}

We consider a gas of point particles that interact with a pair potential. We work in the grand-canonical ensemble where the parameters are the fugacity $z$ and the inverse temperature $\beta$ (both are real and positive numbers). The set $\bbX$ is an open bounded subset of $\bbR^d$ and $\mu(x) = z\dd x$ with $\dd x$ the Lebesgue measure. We actually write $\Lambda=\bbX$ so as to have more traditional notation. The interaction is given by a function $U : \bbR^d \to \bbR \cup \{\infty\}$ which we take to be piecewise continuous; $u(x,y) = \beta U(x-y)$. We suppose that $U$ is stable, i.e.\ that there exists a constant $B\geq0$ such that for any $n$ and any $x_1,\dots,x_n \in \bbR^d$:
\be
\label{U stable}
\sum_{1\leq i<j\leq n} U(x_i-x_j) \geq -Bn.
\ee
Our Assumption \ref{ass stability} holds with $b(x) \equiv \beta B$. The system is translation invariant so all $x \in \bbR^d$ are equivalent. The function of Assumptions \ref{ass KP crit} and \ref{ass tree crit} can then be taken to be a constant, $a(x) \equiv a$. We seek a condition that does not depend on the size of the system. Then integrals over $y$ are on $\bbR^d$ instead of $\Lambda$. By translation invariance we can take $x=0$.

Assumption \ref{ass KP crit} gives the condition
\be
\label{premiere condition suffisante}
z \e{2\beta B} \int_{\bbR^d} \bigl| \e{-\beta U(y)} - 1 \bigr| \dd y \leq a \e{-a}.
\ee
We obviously choose the constant $a$ that maximizes the right side, which is $a=1$. This condition is the one in \cite{Rue}. Let us now assume that $U$ consists of a hard core of radius $r$ and that it is otherwise integrable. Again with $a=1$, Assumption \ref{ass tree crit} gives the condition
\be
\label{deuxieme condition suffisante}
z \e{\beta B} \Bigl[ |\bbB| r^d + \beta \int_{|y|>r} |U(y)| \dd y \Bigr] \leq \e{-1}.
\ee
Here, $|\bbB| = \pi^{d/2}/\Gamma(\frac d2+1)$ is the volume of the ball in $d$ dimensions.
This condition is often better than \eqref{premiere condition suffisante}. Without hard core it is the one in \cite{BrF}. The domains of parameters where these conditions hold correspond to low fugacities and high temperatures.

The thermodynamic pressure is defined as the infinite volume limit of
\be
p_\Lambda(\beta,z) = \frac1{|\Lambda|} \log Z.
\ee
Using Theorem \ref{thm clexp}, we have
\be
p_\Lambda(\beta,z) = \frac1{|\Lambda|} \int_\Lambda \dd x_1 \biggl[ \sum_{n\geq1} z^n \int_\Lambda \dd x_2 \dots \int_\Lambda \dd x_n \varphi(x_1,\dots,x_n) \biggr]
\ee
Consider now any sequence of increasing domains $\Lambda_1 \subset \Lambda_2 \subset \dots$ such that $\Lambda_n \to \bbR^d$. Thanks to the estimate \eqref{the bound}, and using translation invariance, we get
\be
p(\beta,z) \equiv \lim_{n\to\infty} p_{\Lambda_n}(\beta,z) = \sum_{n\geq1} z^n \int_{\bbR^d} \dd x_2 \dots \int_{\bbR^d} \dd x_n \, \varphi(0,x_2,\dots,x_n).
\ee
(The term with $n=1$ is equal to $z$.) This expression for the infinite volume pressure $p(\beta,z)$ should be viewed as a convergent series of analytic functions of $\beta,z$. Then $p(\beta,z)$ is analytic in $\beta$ and $z$ by Vitali theorem and no phase transition takes place in the domain of parameters where the cluster expansion is convergent.

The truncated two-point correlation function $\sigma(x)$ is given by $\hat Z(0,x)$. We consider only the case of Assumption \ref{ass KP crit} but a similar claim can be obtained with Assumption \ref{ass tree crit}. Let $c(x)$ be a function that satisfies the triangle inequality. The estimate of Theorem \ref{thm decay correlation functions} yields
\be
\e{c(x)} \sigma(x) \leq \e{2+4\beta B} \sum_{m\geq0} z^m \e{m+2\beta Bm} \Bigl( \e{c(\cdot)} \bigl| \e{-\beta U(\cdot)} - 1 \bigr| \Bigr)^{*m}(x)
\ee
(with $f^{*0} \equiv f$). Recall that $\|f^{*n}\|_\infty \leq \|f\|_\infty \|f\|_1^{n-1}$, and let
\be
C_p = \bigl\| \e{c(\cdot)} \bigl| \e{-\beta U(\cdot)} - 1 \bigr| \bigr\|_p.
\ee
Then we get
\be
\sigma(x) \leq \e{-c(x)} \e{2+4\beta B} \tfrac{C_\infty}{C_1} \bigl( 1 - z \e{1+2\beta B} C_1 \bigr)^{-1}.
\ee
If the inequality \eqref{premiere condition suffisante} is strict, one can usually find a function $c(x)$ that satisfies the triangle inequality and such that $C_1 \leq (z \e{1+2\beta B})^{-1}$; the truncated two-point correlation function then decays faster than $\e{-c(x)}$.

\section{Polymer systems}
\label{sec class polymers}

Polymer systems are discrete, which is technically simpler, but they also have internal structure. The first application of cluster expansions to polymer systems is due to Gruber and Kunz \cite{GK}. Among the many articles devoted to this subject, let us mention \cite{KP,Dob,FP}. The main goal of this section is to illustrate our setting; we therefore restrict ourselves to a specific model of polymers with both repulsive and attractive interactions.

Our space $\bbX$ is the set of all finite connected subsets of $\bbZ^d$. The measure $\mu$ is the counting measure multiplied by the {\it activity} $z(x)$ (a function $\bbX \to \bbR_+$). We choose $z(x) = \e{-\gamma|x|}$ with $\gamma>0$. The interaction is hard core when polymers overlap and it is attractive when they touch:
\be
u(x,y) = \begin{cases} \infty & \text{if } x \cap y \neq \emptyset, \\ -\eta \, c(x,y) & \text{if } x \cap y = \emptyset. \end{cases}
\ee
Here, $c(x,y)$ is the number of ``contacts" between $x$ and $y$, i.e.\ the number of bonds between sites of $x$ and $y$; $\eta>0$ is a parameter. The interaction is zero when the distance between polymers is greater than 1.

The stability condition can be written
\be
\tfrac12 \sum_{i=1}^n \sum_{j\neq i} u(x_i,x_j) \geq -\sum_{i=1}^n b(x_i).
\ee
Only disjoint polymers need to be considered, the left side is infinite otherwise. The sum over $j$ is always larger than $-\eta$ times the number of bonds connecting $x_j$ with its exterior. Thus we can take $b(x) = \eta d |x|$.

The function $a$ in Assumption \ref{ass tree crit} grows like $|x|$, so it is natural to choose $a(x) = a|x|$ for some constant $a$. A sufficient condition is that
\be
\sum_{y, y \cap x \neq \emptyset} z(y) \e{a|y| + \eta d|y|} + \sum_{y, \dist(x,y) = 1} \eta z(y) c(x,y) \e{a|y| + \eta d|y|} \leq a|x|.
\ee
We can bound $\eta c(x,y)$ by $2\eta d|y|$. Summing over the sites of $x$, and requiring that $y$ contains the given site or comes at distance 1, we get
\be
\sum_{y \ni 0} (1+2d\eta |y|) z(y) \e{a|y| + \eta d|y|} \leq a.
\ee
We used the fact that the activity is translation invariant. If $x$ is a connected set, there exists a closed walk with nearest neighbor jumps whose support is $x$, and whose length is at most $2(|x|-1)$. This can be seen by induction: knowing the walk for $x$, it is easy to construct one for $x \cup \{i\}$. The number of connected sets of cardinality $n$ that contain the origin is therefore smaller than the number of walks of length $2n-3$ starting at the origin, which is equal to $(2d)^{2n-3}$. Then it suffices that
\be
\sum_{n\geq1} \e{-n (\gamma - a - 3d\eta - 2\log 2d)} \leq (2d)^3a.
\ee
This is equivalent to
\be
\gamma \geq a + \log\bigl( 1 + \tfrac1{(2d)^3 a} \bigr) + 3d\eta + 2\log 2d.
\ee
Assumption \ref{ass tree crit} holds for any $a$. Using $\log(1+t) \leq t$ and optimizing on $a$, we find the sufficient condition
\be
\gamma \geq 2 (2d)^{-3/2} + 3d\eta + 2\log 2d
\ee
with $a = (2d)^{-3/2}$.

We have just established the existence of a low density phase provided the activity is small enough. The condition depends on the contact parameter $\eta$. For large $\eta$ one should expect interesting phases with many contacts between the polymers.

\section{Quantum gas}
\label{sec quantum gas}

We follow a course that is similar to Ginibre \cite{Gin}, using the Feynman-Kac formula so as to get a gas of winding Brownian loops. We get comparable results, with a larger domain of convergence in the case of integrable potentials. Winding Brownian loops are kind of continuous polymers; they combine the difficulties of both cases above --- the continuous nature and the internal structure.

\subsection{Feynman-Kac representation}

The state space for $N$ fermions (resp.\ bosons) in a domain $\Lambda \subset \bbR^d$ is the Hilbert space $L^2_{\rm anti}(\Lambda^N)$ (resp.\ $L^2_{\rm sym}(\Lambda^N)$) of square-integrable complex functions that are antisymmetric (resp.\ symmetric) with respect to their arguments. The Hamiltonian is
\be
\label{quantum Ham}
H = -\sum_{i=1}^N \Delta_i + \sum_{1\leq i<j \leq N} U(q_i-q_j),
\ee
with $\Delta_i$ the Laplacian for the $i$-th variable and $U(q)$ a multiplication operator. As in the classical case, we consider the grand-canonical ensemble whose parameters are the fugacity $z$ and the inverse temperature $\beta$. The partition function is
\be
\label{quantum part fct}
Z = \sum_{N\geq0} z^N \Tr \e{-\beta H}.
\ee

We need to cast the partition function in the form \eqref{def Z}, which can be done using the Feynman-Kac representation. Namely, we have
\bm
\label{FK rep}
Z = \sum_{N\geq0} \frac{z^N}{N!} \sum_{\pi\in\caS_N} \varepsilon(\pi) \int_{\Lambda^N} \dd q_1 \dots \dd q_N  \int\dd W_{q_1,q_{\pi(1)}}^{2\beta}(\omega_1) \dots \int\dd W_{q_N,q_{\pi(N)}}^{2\beta}(\omega_N)  \\\Bigl( \prod_{i=1}^N \upchi_\Lambda(\omega_i) \Bigr) \exp\Bigl\{ -\tfrac12 \sum_{1\leq i<j\leq N} \int_0^{2\beta} U \bigl( \omega_i(s)-\omega_j(s) \bigr) \dd s \Bigr\}.
\end{multline}
$\caS_N$ is the symmetric group of $N$ elements; $\varepsilon(\pi)$ is equal to the signature of the permutation $\pi$ for fermions, $\varepsilon(\pi) \equiv 1$ for bosons; $W_{q,q'}^t$ is the Wiener measure for the Brownian bridge from $q$ to $q'$ in time $t$ --- the normalization is chosen so that
\be
\int\dd W_{q,q'}^t(\omega) = (2\pi t)^{-d/2} \e{-|q-q'|^2/2t};
\ee
$\upchi_\Lambda(\omega)$ is one if $\omega(s) \in \Lambda$ for all $0 \leq s \leq 2\beta$, it is zero otherwise. An introduction to the Feynman-Kac formula in this context can be found in the survey of Ginibre \cite{Gin}.

The right side of Eq.\ \eqref{FK rep} is well defined for a large class of functions $U : \bbR^d \to \bbR \cup \{\infty\}$, that includes all piecewise continuous functions. Thus we take \eqref{FK rep} as the definition for $Z$. Under additional assumptions on $U$, \eqref{FK rep} is equal to \eqref{quantum part fct} with Hamiltonian \eqref{quantum Ham} and with Dirichlet boundary conditions.

We now rewrite the grand-canonical partition function in terms of winding loops. Let $\bbX_k$ be the set of continuous paths $[0,2\beta k] \to \bbR^d$ that are closed. Its elements are denoted $x = (q,k,\omega)$, with $q \in \bbR^d$ the starting point, $k$ the winding number, and $\omega$ the path; we have $\omega(0) = \omega(2\beta k) = q$. We consider the measure $\mu$ given by
\be
\mu(\dd x) = \frac{\varepsilon^{k+1} z^k}k \dd q \e{-v(x)} \upchi_\Lambda(\omega) W_{q,q}^{2\beta k}(\dd\omega).
\ee
Here, $v(x)$ is a self-interaction term that is defined below in \eqref{def uv}; $\varepsilon=-1$ for fermions and $+1$ for bosons. Let $\bbX = \cup_{k\geq1} \bbX_k$; the measure $\mu$ naturally extends to a measure on $\bbX$. The grand-canonical partition function can then be written as
\be
Z = \sum_{n\geq0} \frac1{n!} \int_{\bbX^n} \dd\mu(x_1) \dots \dd\mu(x_n) \exp\Bigl\{ -\sum_{1\leq i<j\leq n} u(x_i,x_j) \Bigr\}.
\ee
Let $x=(q,k,\omega)$ and $x'=(q',k',\omega')$. The self-interaction $v(x)$ and the 2-loop interaction $u(x,x')$ are given by
\be
\label{def uv}
\begin{split}
&v(x) = \sum_{0\leq\ell<m\leq k-1} \tfrac12 \int_0^{2\beta} U \bigl( \omega(s+2\beta\ell) - \omega(s+2\beta m) \bigr) \dd s, \\
&u(x,x') = \sum_{\ell=0}^{k-1} \sum_{\ell'=0}^{k'-1} \tfrac12 \int_0^{2\beta} U \bigl( \omega(s+2\beta\ell) - \omega'(s+2\beta\ell') \bigr) \dd s.
\end{split}
\ee

We now treat separately the case of integrable potentials and the case of more general potentials.

\subsection{Stable integrable potentials}
\label{sec stable integrable}

We suppose that $U$ is stable with constant $B$, i.e.\ it satisfies Eq.\ \eqref{U stable}. For given loops $x_1,\dots,x_n$, stability implies that
\be
\sum_{i=1}^n v(x_i) + \sum_{1\leq i<j\leq n} u(x_i,x_j) \geq -\beta B \sum_{i=1}^n k_i.
\ee
Then Assumption \ref{ass stability} holds with
\be
\label{stability condition}
b(x) = \beta B k + v(x).
\ee
Notice that $b(x) \geq 0$, again by the stability of $U$.

We use Assumption \ref{ass tree crit} and we choose a function $a(x) = ak$ with a constant $a$ to be determined later. Explicitly, the assumption is that for any $x=(q,k,\omega)$
\be
\label{explicit condition}
\sum_{k'\geq1} \frac{z^{k'} \e{ak'}}{k'} \int_{\bbR^d} \dd q' \int\dd W_{q',q'}^{2k'\beta}(\omega') \e{-v(x')+b(x')} |u(x,x')| \leq ak.
\ee
We have lifted the restriction that $\omega'(s) \in \Lambda$ because we want a condition that does not depend on $\Lambda$. Eq.\ \eqref{explicit condition} is easier to handle than its appearance suggests. Notice that the term in the second exponential is just $\beta Bk'$. Using $\int\dd W_{qq}^t(\omega) f(\omega) = \int\dd W_{00}^t(\omega) f(\omega+q)$, and the definition \eqref{def uv} of $u(x,x')$, it is enough that
\bm
\sum_{\ell=0}^{k-1} \tfrac12 \int_0^{2\beta} \dd s \sum_{k'\geq1} \frac{z^{k'} \e{(a+\beta B)k'}}{k'} \sum_{\ell'=0}^{k'-1} \int_{\bbR^d} \dd q' \\
\int\dd W_{0,0}^{2k'\beta}(\omega') \bigl| U(\omega(s+2\beta\ell)-\omega'(s+2\beta\ell')-q') \bigr| \leq ak.
\end{multline}
We can immediately integrate over $q'$, which yields $\|U\|_1$. The Wiener integral then gives $(4\pi k'\beta)^{-d/2}$ and we get the equivalent condition
\be
\frac{\beta \|U\|_1}{(4\pi\beta)^{d/2}} \sum_{k'\geq1} \frac{z^{k'} \e{(a+\beta B)k'}}{(k')^{d/2}} \leq a.
\ee
For any $d,\beta,\|U\|_1,a$, the inequality holds for $z$ small enough. Notice that $z<1$ in any case. One can get a more explicit condition for $d\geq3$ by choosing $a$ such that $z \e{a+\beta B}$ = 1. This yields
\be
\label{condition sur z}
z \leq \exp\Bigl\{ -\beta \Bigl[ \frac{\|U\|_1 \zeta(\frac d2)}{(4\pi\beta)^{d/2}} + B \Bigr] \Bigr\}.
\ee
Here, $\zeta(\frac d2) = \sum_{n\geq1} n^{-\frac d2}$ is the Riemann zeta function.

When $d=3$ and when the potential is repulsive, one can rewrite \eqref{condition sur z} in a more transparent way. Let $a_0 = \frac1{8\pi} \|U\|_1$ denote the Born approximation to the scattering length. The condition is then
\be
z \leq \exp\bigl\{ -\tfrac{\zeta(\frac32)}{\sqrt\pi} \tfrac{a_0}{\sqrt\beta} \bigr\}.
\ee
The critical fugacity is expected to be greater than 1. The present result helps nonetheless to obtain a range of densities where the pressure is analytic. In the bosonic case it compares well with physicists' expectations \cite{SU}.

\subsection{Stable potentials with hard core}
\label{sec stable hard core}

The presence of a hard core makes the situation more complicated; we only sketch the argument in this section without trying to get explicit bounds. Our aim is to show that, using Theorem \ref{thm clexp}, the problem of convergence of the cluster expansion reduces to estimates of Wiener sausages.

We consider an interaction $U = U' + U''$. We assume that $U' \geq 0$ is a repulsive potential of radius $r$, with a hard core of radius $0< r' \leq r$, that $U''(q)=0$ for $|q|<r$, and that $U''$ is integrable otherwise. We suppose that the stability condition takes the slightly stronger form
\be
\sum_{i=1}^n U(q_0-q_i) \geq -B
\ee
for any $q_0,\dots,q_n$ such that $|q_i-q_j|>r'$. For potentials with a hard core this is equivalent to the property \eqref{U stable}, possibly with a different constant $B$. Then one has \cite{Gin1}
\be
\sum_{i=1}^n v(x_i) + 2 \sum_{1\leq i<j \leq n} u(x_i,x_j) \geq -2\beta B \sum_{i=1}^n k_i.
\ee
Then Assumption \ref{ass stability} holds with
\be
\label{new stability condition}
b(x) = \beta B k + \tfrac12 v(x).
\ee
Notice that the stability condition also holds with $b$ given in \eqref{stability condition} (and with a better constant $B$). The advantage of \eqref{new stability condition} is the factor $\frac12$ in front of $v(x)$. Then $\e{2b(x)}$ involves a term that cancels the self-interactions of $x$.

Given a loop $x=(q,k,\omega)$, let $S(x)$ be the Wiener sausage generated by a ball of radius $r$ when its center moves along the trajectory $\omega$:
\be
S(x) = \bigl\{ q \in \bbR^d : |\omega(s)-q| \leq r \text{ for some } s \in [0,2k\beta] \bigr\}.
\ee
We also define the Wiener sausage that corresponds to the difference of two winding loops $x = (q,k,\omega)$ and $x' = (q',k',\omega')$:
\bm
S(x,x') = \bigl\{ q'' : |\omega(s+2\beta\ell) - \omega'(s+2\beta\ell') - q''| \leq r \\
\text{ for some } \ell = 0,\dots,k-1; \ell' = 0,\dots,k'-1; s \in [0,2\beta] \bigr\}.
\end{multline}
We denote the volume of a Wiener sausage $S(\cdot)$ by $|S(\cdot)|$. One can check that
\be
\label{Ginibre bound}
|S(x,y)| \leq \frac{|S(x)| |S(y)|}{r^d |\bbB|}
\ee
with $|\bbB|$ the volume of the unit ball (see Appendix 2 in \cite{Gin1}).

We choose $a(x) = |S(x)| + k$ in Assumption \ref{ass KP crit}. Then a sufficient condition is that for any $x \in \bbX$,
\be
\label{ca suffit}
\sum_{k'\geq1} \frac{z^{k'} \e{(2\beta B + 1) k'}}{k'} \int_{\bbR^d} \dd q' \int\dd W^{2k'\beta}_{0,0}(\omega') |\zeta(x,x'+q')| \e{|S(x')|} \leq |S(x)| + k.
\ee
We consider separately the cases where $q'$ belongs or not to $S(x,x')$. First,
\be
\int_{S(x,x')} |\zeta(x,x'+q')| \, \dd q' \leq |S(x,x')|,
\ee
which we bound using \eqref{Ginibre bound}. Second, using $|\e{-u(x,x')}-1| \leq |u(x,x')| \e{\beta Bk'}$,
\be
\int_{\bbR^d \setminus S(x,x')} |\zeta(x,x'+q')| \, \dd q' \leq \beta \|U''\|_1 k k' \e{\beta B k'}.
\ee
We certainly get \eqref{ca suffit} if we have the two inequalities
\be
\begin{split}
&\sum_{k'\geq1} \frac{z^{k'} \e{(2\beta B + 1) k'}}{k'} \int\dd W^{2k'\beta}_{0,0}(\omega') \e{|S(x')|} |S(x')| \leq r^d |\bbB|, \\
&\sum_{k'\geq1} z^{k'} \e{(3\beta B + 1) k'} \int\dd W^{2k'\beta}_{0,0}(\omega') \e{|S(x')|}  \leq \frac1{\beta \|U''\|_1}.
\end{split}
\ee
One can estimate the integrals of Wiener sausages, see \cite{Gin1}, so that both conditions hold if $z$ is small enough.

Now that the cluster expansion is known to converge, it is possible to write the pressure as an absolutely convergent series of analytic functions in $\beta$ and $z$. It is also possible to study the decay of correlation functions. In the case of potentials that consist of hard core plus integrable part, one can apply Assumption 2' instead. This may give better results, especially if the integrable part is mostly attractive.

\section{Tree estimates}
\label{sec tree estimate}

In this section we obtain estimates of sums of connected graphs in terms of sums of trees. Our main result is Proposition \ref{prop tree estimate} below. Such estimates seem to have been introduced by Penrose \cite{Pen} and they have often been considered  in the past \cite{BrF,BaF,Bry,BK,Pfi,AR,Far}. We introduce a minimal setting that clarifies its r\^ole in the cluster expansion. Namely, we fix the polymers so we only deal with the numbers that represent their interactions, $\zeta$ or $u$, and the stability function $b$. Assumption \ref{ass stability} is vital here, but Assumptions \ref{ass KP crit} and \ref{ass tree crit} are not used in this section.

Let $\caT_n \subset \caC_n$ denote the set of trees with $n$ vertices. Let $n$ be an integer, $b_1,\dots,b_n$ be real nonnegative numbers, and $\zeta_{ij}=\zeta_{ji}$, $1\leq i,j \leq n$, be complex numbers. We assume that the following bound holds for any subset $I \subset \{1,\dots,n\}$:
\be
\label{stability}
\prod_{i,j \in I, i<j} |1+\zeta_{ij}| \leq \prod_{i\in I} \e{b_i},
\ee
Let $u_{ij}$ be such that $\zeta_{ij} = \e{-u_{ij}}-1$; let $\overline u_{ij} = 1$ if $\zeta_{ij}=-1$, and $\overline u_{ij} = u_{ij}$ otherwise. We state two distinct tree estimates, the first one involving $|\zeta_{ij}|$ and the second one involving $|\overline u_{ij}|$. These bounds will allow to prove the convergence under either Assumption \ref{ass KP crit} or Assumption \ref{ass tree crit}.

\begin{proposition}
\label{prop tree estimate}
If \eqref{stability} holds true, we have the two bounds
\begin{itemize}
\item[(a)] $\displaystyle \Bigl| \sum_{G \in \caC_n} \prod_{\{i,j\} \in G} \zeta_{ij} \Bigr| \leq  \Bigl( \prod_{i=1}^n \e{2b_i} \Bigr) \sum_{G \in \caT_n} \prod_{\{i,j\} \in G} |\zeta_{ij}|$;
\item[(b)] $\displaystyle \Bigl| \sum_{G \in \caC_n} \prod_{\{i,j\} \in G} \zeta_{ij} \Bigr| \leq  \Bigl( \prod_{i=1}^n \e{b_i} \Bigr) \sum_{G \in \caT_n} \prod_{\{i,j\} \in G} |\overline u_{ij}|$.
\end{itemize}
\end{proposition}

We actually conjecture that the following estimate holds under the same hypotheses:
\be
\Bigl| \sum_{G \in \caC_n} \prod_{\{i,j\} \in G} \zeta_{ij} \Bigr| \leq  \Bigl( \prod_{i=1}^n \e{b_i} \Bigr) \sum_{G \in \caT_n} \prod_{\{i,j\} \in G} \min(|u_{ij}|, |\zeta_{ij}|).
\ee

We prove Proposition \ref{prop tree estimate} (a) below using Ruelle's algebraic approach \cite{Rue}. This method is usually combined with a Banach fixed point argument for correlation functions. However, we use it differently so as to get a tree estimate. Proposition \ref{prop tree estimate} (b) follows from a tree identity due to Brydges, Battle, and Federbush \cite{BrF,BaF,Bry}, combined with an argument due to Procacci \cite{Pro}; its proof can be found at the end of this section.

Let $\caA$ be the set of complex functions on the power set $\caP(\{1,\dots,n\})$. We introduce the following multiplication operation for $f,g \in \caA$:
\be
f*g(I) = \sum_{J \subset I} f(J) g(I \setminus J).
\ee
We use the standard conventions for sums and products, namely that the empty sum is zero and the empty product is one. Together with the addition, $\caA$ is a commutative algebra with unit $1_\caA(I) = \delta_{I,\emptyset}$. It is possible to check that each $f \in \caA \setminus \{0\}$ has a unique inverse, which we denote $f^{*(-1)}$. We have
\be
f^{*k}(I) = \sumtwo{J_1,\dots,J_k \subset I}{J_i \cap J_j = \emptyset, \cup J_i = I} f(J_1) \dots f(J_k).
\ee
Let $\caA_0$ be the subset of functions $f$ such that $f(\emptyset)=0$ ($\caA_0$ is an ideal of $\caA$). Notice that $f^{*k}=0$ for any $k>n$ , when $f \in \caA_0$. We define the exponential mapping $\exp_\caA : \caA_0 \to \caA_0 + 1_\caA$ by
\be
\exp_\caA f = 1_\caA + f + \tfrac12 f^{*2} + \dots + \tfrac1{n!} f^{*n}.
\ee

Let $\Phi$ and $\Psi$ be the functions defined by
\be
\label{def Phi Psi}
\begin{split}
&\Phi(I) = \sum_{G \in \caC(I)} \prod_{\{i,j\} \in G} \zeta_{ij}, \\
&\Psi(I) = \prod_{i,j \in I, i<j} (1+\zeta_{ij}) = \sum_{G \in \caG(I)} \prod_{\{i,j\} \in G} \zeta_{ij}.
\end{split}
\ee
Here, $\caG(I)$ (resp.\ $\caC(I)$) is the set of graphs (resp.\ connected graphs) on $I$. We have the relation
\be
\Psi = \exp_\caA \Phi.
\ee

We also introduce an operation that is reminiscent of differentiation:
\be
D_J f(I) = \begin{cases} f(I \cup J) & \text{if } I \cap J = \emptyset, \\ 0 & \text{otherwise.} \end{cases}
\ee
One can check that $D_{\{i\}} \exp_\caA f = (\exp_\caA f) * D_{\{i\}} f$.

For disjoint $I,J \subset \{1,\dots,n\}$, we define
\be
\label{def f}
g(I,J) = \bigl( \Psi^{*(-1)} * D_I \Psi \Big) (J).
\ee

Let $I \subset \{1,\dots,n\}$. The assumption of Proposition \ref{prop tree estimate} implies that
\be
\prod_{i\in I} \prod_{j\in I\setminus\{i\}} |1 + \zeta_{ij}| \leq \prod_{i \in I} \e{2b_i}.
\ee
Then there exists $i \in I$ such that
\be
\prod_{j\in I\setminus\{i\}} |1 + \zeta_{ij}| \leq \e{2b_i}.
\ee
Such $i$ is not unique in general but it does not matter. We consider a function $\iota$ that assigns one of the indices $i = \iota(I)$ above to each nonempty subset $I \subset \{1,\dots,n\}$. Notice that $\iota(I) \in I$ for any subset $I$. It is also useful to introduce the notation $I' = I \setminus \{\iota(I)\}$.

\begin{lemma}
\label{lem KS equation}
The function $g(I,J)$ of Eq.\ \eqref{def f} is solution of the following equation.
\[
\begin{cases} g(\emptyset,J) = \delta_{\emptyset,J}, & \\ \displaystyle g(I,J) = \Bigl( \prod_{i \in I'} (1 + \zeta_{i,\iota(I)}) \Bigr) \sum_{K \subset J} \Bigl( \prod_{i \in K} \zeta_{i,\iota(I)} \Bigr) g(I' \cup K, J \setminus K) & \text{if } I \neq \emptyset. \end{cases}
\]
\end{lemma}

Since the equation gives $g(I,J)$ in terms of $g(K,L)$ with $|K|+|L| = |I| + |J| - 1$, it is well defined inductively and it has a unique solution. Notice that $g(\emptyset,\emptyset)=1$, and that $g(\{i\},\emptyset) = 1$ for any index $i$.

\begin{proof}
Recall the definition \eqref{def Phi Psi} of $\Psi$. For disjoint $I,K$ we have
\be
\begin{split}
\Psi(I \cup K) &=\Bigl( \prod_{j \in I' \cup K} (1+\zeta_{j,\iota(I)}) \Bigr) \Psi(I' \cup K) \\
&= \Bigl( \prod_{j \in I'} (1+\zeta_{j,\iota(I)}) \Bigr) \Bigl( \sum_{L \subset K} \prod_{k \in L} \zeta_{k,\iota(I)} \Bigr) \Psi(I' \cup K).
\end{split}
\ee
Then
\be
\begin{split}
g(I,J) &= \sum_{K \subset J} \Psi^{*(-1)}(J \setminus K) \Psi(I \cup K) \\
&= \Bigl( \prod_{j \in I'} (1+\zeta_{j,\iota(I)}) \Bigr) \sum_{L \subset K \subset J} \bigl( \prod_{k \in L} \zeta_{k,\iota(I)} \Bigr) \Psi^{*(-1)}(J \setminus K) \Psi(I' \cup K) \\
&= \Bigl( \prod_{j \in I'} (1+\zeta_{j,\iota(I)}) \Bigr) \sum_{L \subset J} \bigl( \prod_{k \in L} \zeta_{k,\iota(I)} \Bigr) \sum_{K' \subset J \setminus L} \Psi^{*(-1)}(J \setminus L \setminus K') \Psi(I' \cup L \cup K').
\end{split}
\ee
The last sum is equal to $g(I' \cup L, J \setminus L)$. One recognizes the equation of Lemma \ref{lem KS equation}.
\end{proof}

We now estimate the function $g$ using another function $h$ that satisfies an equation that is similar to that of Lemma \ref{lem KS equation}.
\be
\label{KS bound}
\begin{cases} h(\emptyset,J) = \delta_{\emptyset,J}, & \\ \displaystyle h(I,J) = \e{2b_{\iota(I)}} \sum_{K \subset J} \Bigl( \prod_{i \in K} |\zeta_{i,\iota(I)}| \Bigr) h(I' \cup K, J \setminus K) & \text{if } I \neq \emptyset. \end{cases}
\ee
It also has a unique solution. Since $\prod_{i \in I'} |1 + \zeta_{i,\iota(I)}| \leq \e{2b_{\iota(I)}}$, we can check inductively that
\be
\label{I bounded by J}
|g(I,J)| \leq h(I,J)
\ee
for any sets $I,J$ (with $I \cap J = \emptyset$). Now the function $h$ can be written explicitly \cite{MP,PZ}. Let $\caF_I(J)$ be the set of forests on $I \cup J$ rooted in $I$. That is, a graph $G \in \caF_I(J)$ is a forest such that each tree contains exactly one element of $I$.

\begin{lemma}
The solution of Eq.\ \eqref{KS bound} is
\[
h(I,J) = \Bigl( \prod_{i \in I \cup J} \e{2b_i} \Bigr) \sum_{G \in \caF_I(J)} \prod_{\{i,j\} \in G} |\zeta_{ij}|.
\]
\end{lemma}

\begin{proof}
Since the solution to Eq.\ \eqref{KS bound} is unique, it is enough to check that the Ansatz of the lemma satisfies the equation. First, let us observe that both sides are multiplied by $\prod_{i \in I \cup J} \e{2b_i}$. Thus it is enough to consider the case $b_i \equiv 0$.

The sum over graphs in $\caF_I(J)$ can be realized by first summing over the set $K$ of indices (necessarily in $J$) that are connected to $\iota(I)$; then over sets of trees in $J \setminus K$, and over connections to $I' \cup K$. Explicitly,
\be
\sum_{G \in \caF_I(J)} \prod_{\{i,j\}} |\zeta_{ij}| = \sum_{K \subset J} \Bigl( \prod_{i \in K} |\zeta_{i,\iota(I)}| \Bigr) \sum_{G \in \caF_{I' \cup K}(J \setminus K)} \prod_{\{i,j\} \in G} |\zeta_{ij}|.
\ee
This equation is precisely \eqref{KS bound}.
\end{proof}

\begin{proof}[Proof of Proposition \ref{prop tree estimate} (a)]
When $I$ has a single element, the function $g$ is equal to
\be
g \bigl( \{1\},\{2,\dots,n\} \bigr) = \sum_{G \in \caC_n} \prod_{\{i,j\} \in G} \zeta_{ij}.
\ee
This is the left side of Proposition \ref{prop tree estimate} (a). We have $\caF_{\{1\}}(\{2,\dots,n\}) = \caT_n$, the set of trees with $n$ vertices. Thus $h(\{1\},\{2,\dots,n\})$ is equal to the right side of Proposition \ref{prop tree estimate} (a), and the proof follows from Eq.\ \eqref{I bounded by J}.
\end{proof}

We now turn to the proof of Proposition \ref{prop tree estimate} (b). Notice that in absence of ``hard cores", i.e.\ when $\zeta_{ij} \neq -1$, our claim is just a reformulation of Corollary 3.2 (a) of \cite{Bry}. The present proof follows \cite{Pro}.

\begin{proof}[Proof of Proposition \ref{prop tree estimate} (b)]
Let $P$ be the set of $\{i,j\}$ such that $\zeta_{ij}=-1$, i.e.\ such that $\Re u_{ij} = \infty$. We regularize those numbers by setting
\be
v_{ij}^{(m)} = \begin{cases} m & \text{if } \{i,j\} \in P, \\ u_{ij} & \text{if } \{i,j\} \notin P. \end{cases}
\ee
This allows to use the tree identity of \cite{BrF,BaF,Bry}; we will eventually take $m$ to infinity. The tree identity can be written
\be
\label{tree identity}
\sum_{G \in \caC_n} \prod_{\{i,j\} \in G} \bigl( \e{-v_{ij}^{(m)}} - 1 \bigr) = \sum_{G \in \caT_n} \prod_{\{i,j\} \in G} (-v_{ij}^{(m)}) \int\dd\lambda_G(\{s_{ij}\}) \e{-\sum_{i<j} s_{ij} v_{ij}^{(m)}}.
\ee
The full definition of the measure $\lambda_G$ can be found in \cite{Bry}; here we only mention its relevant properties. $\lambda_G$ depends on the tree $G$ and it is a probability measure supported on a subset $L \subset \{ s_{ij} \in [0,1] : 1 \leq i<j \leq n \}$. If $v_{ij}^{(m)}$ satisfies the stability condition, then
\be
\sum_{i,j \in I, i<j} s_{ij} \Re u_{ij} \geq -\sum_{i \in I} b_i,
\ee
for all $\{s_{ij}\} \in L$ and all $I \subset \{1,\dots,n\}$.

If $\Re u_{ij} \neq \infty$ for any $i,j$, we immediately get the claim by using the stability condition above, since $\int\dd\lambda_G(\{s_{ij}\}) = 1$. The extension to possibly infinite numbers can be obtained using a trick due to Procacci \cite{Pro}. Let $G$ be a fixed tree and $\varepsilon>0$. We introduce
\be
w_{ij}^{(m,\varepsilon)} = \begin{cases} (1-\varepsilon)m & \text{if } \{i,j\} \in G \cap P, \\ \varepsilon & \text{if } \{i,j\} \in G \setminus P, \\ 0 & \text{otherwise.} \end{cases}
\ee
Then
\be
\sum_{1\leq i<j\leq n} s_{ij} \Re v_{ij}^{(m)} \geq \sum_{1\leq i<j\leq n} s_{ij} \Re v_{ij}^{(\varepsilon m)} + \sum_{1\leq i<j\leq n} s_{ij} w_{ij}^{(m,\varepsilon)} - \varepsilon |G \setminus P|.
\ee
If $m$ is large enough (depending on $\varepsilon$), $\{ v_{ij}^{(\varepsilon m)} \}$ is stable and the first term of the right side is bounded below by $-\sum b_i$. Then
\bm
\biggl| \prod_{\{i,j\} \in G} (-v_{ij}^{(m)}) \int\dd\lambda_G(\{s_{ij}\}) \e{-\sum s_{ij} v_{ij}^{(m)}} \biggr| \leq \biggl( \prod_{i=1}^n \e{b_i} \biggr) \biggl( \prod_{\{i,j\} \in G \cap P} \frac{w_{ij}^{(m,\varepsilon)}}{1-\varepsilon} \biggr) \\
\times \biggl( \prod_{\{i,j\} \in G \setminus P} \frac{w_{ij}^{(m,\varepsilon)} \e\varepsilon}\varepsilon |u_{ij}| \biggr) \int\dd\lambda_G(\{s_{ij}\}) \e{-\sum_{i<j} s_{ij} w_{ij}^{(m,\varepsilon)}}.
\end{multline}
A special case of the tree identity \eqref{tree identity} is
\be
\begin{split}
\prod_{\{i,j\} \in G} w_{ij}^{(m,\varepsilon)} \int\dd\lambda_G(\{s_{ij}\}) &\e{-\sum_{i<j} s_{ij} w_{ij}^{(m,\varepsilon)}} = \prod_{\{i,j\} \in G} \bigl( 1 - \e{-w_{ij}^{(m,\varepsilon)}} \bigr) \\
&= (1 - \e{-(1-\varepsilon)m})^{|G \cap P|} (1 - \e{-\varepsilon})^{|G \setminus P|}.
\end{split}
\ee
We get
\bm
\biggl| \prod_{\{i,j\} \in G} (-v_{ij}^{(m)}) \int\dd\lambda_G(\{s_{ij}\}) \e{-\sum s_{ij} v_{ij}^{(m)}} \biggr|
\\ \leq \biggl( \prod_{i=1}^n \e{b_i} \biggr) \biggl( \frac{1 - \e{-(1-\varepsilon)m}}{1-\varepsilon} \biggr)^{|G \cap P|} \biggl( \prod_{\{i,j\} \in G \setminus P} |u_{ij}| \frac{\e\varepsilon - 1}\varepsilon \biggr).
\end{multline}
We can insert this estimate into Eq.\ \eqref{tree identity}. Letting $m\to\infty$ and then $\varepsilon\to0$, we get Proposition \ref{prop tree estimate} (b).
\end{proof}

\section{Proofs of the theorems}
\label{sec proofs}

In this section we prove the theorems of Section \ref{sec clexp}. We consider only the case where Assumption \ref{ass KP crit} holds true --- the case with Assumption \ref{ass tree crit} is entirely the same, one only needs to replace all $|\zeta(x,y)|$ with $|u(x,y)|$ and all $\e{2b(\cdot)}$ with $\e{b(\cdot)}$. The proofs are based on the following tree estimate, which is a direct consequence of Proposition \ref{prop tree estimate}: for almost all $x_1,\dots,x_n \in \bbX$,
\be
\label{tree estimate}
|\varphi(x_1,\dots,x_n)| \leq \frac1{n!} \prod_{i=1}^n \e{2b(x_i)} \sum_{G \in \caT_n} \prod_{\{i,j\} \in G} |\zeta(x_i,x_j)|.
\ee

\begin{proof}[Proof of Theorem \ref{thm clexp}]
We start by proving the bound \eqref{the bound}. Let us introduce
\begin{align}
&K_N(x_1) = \sum_{n=1}^N \frac1{(n-1)!} \int\dd|\mu|(x_2) \dots \int\dd|\mu|(x_n) \prod_{i=1}^n \e{2b(x_i)} \sum_{G \in \caT_n} \prod_{\{i,j\} \in G} |\zeta(x_i,x_j)|, \nn\\
&K(x) = \lim_{N\to\infty} K_N(x).
\label{def K_N}
\end{align}
(The term $n=1$ is equal to $\e{2b(x_1)}$ by definition.) We show by induction that
\be
\label{bound K_N}
K_N(x) \leq \e{a(x)+2b(x)}
\ee
for any $N$. Then $K(x) \leq \e{a(x)+2b(x)}$ for almost all $x$, and using \eqref{tree estimate} we get \eqref{the bound}.

The case $N=1$ reduces to $1 \leq \e{a(x)}$ and it is clear. The sum over trees with $n$ vertices can be written as a sum over forests on $\{2,\dots,n\}$, and a sum over edges between 1 and each tree of the forest. Explicitly,
\bm
K_N(x_1) = \sum_{n=1}^N \frac1{(n-1)!} \sum_{m\geq1} \sumtwo{\{V_1,\dots,V_m\}}{\text{partition of } \{2,\dots,n\}} \int\dd|\mu|(x_2) \dots \int\dd|\mu|(x_n) \\
\prod_{i=1}^n \e{2b(x_i)} \prod_{k=1}^m \biggl( \sum_{\ell \in V_k} |\zeta(x_1,x_\ell)| \sum_{G \in \caT(V_k)} \prod_{\{i,j\} \in G} |\zeta(x_i,x_j)| \biggr).
\end{multline}
Here, $\caT(V)$ denotes the set of trees with $V$ as the set of vertices. If $|V_k|=1$ the sum over $G \in \caT(V_k)$ is one by definition. The term after the sum over partitions depends on the cardinalities of the $V_k$'s, but not on the actual labeling. Also, each $\ell \in V_k$ gives the same contribution. We get
\bm
K_N(x_1) = \e{2b(x_1)} \sum_{n=1}^N \sum_{m\geq1} \frac1{m!} \sumtwo{n_1,\dots,n_m \geq 1}{n_1 + \dots + n_m = n-1} \prod_{k=1}^m \biggl( \frac1{(n_k-1)!} \\
\int\dd|\mu|(y_1) \dots \int\dd|\mu|(y_{n_k}) \, |\zeta(x_1,y_1)| \prod_{i=1}^{n_k} \e{2b(y_i)} \sum_{G \in \caT_{n_k}} \prod_{\{i,j\} \in G} |\zeta(y_i,y_j)| \biggr)
\end{multline}
We obtain an upper bound by releasing the constraint $n_1 + \dots n_m \leq N-1$ to $n_k \leq N-1$, $1 \leq k \leq m$. We then get
\be
\begin{split}
K_N(x_1) &\leq \e{2b(x_1)} \exp\biggl\{ \sum_{n=1}^{N-1} \frac1{(n-1)!} \int\dd|\mu(y_1) \dots \int\dd|\mu|(y_n) \, |\zeta(x_1,y_1)| \\
&\hspace{5cm} \prod_{i=1}^n \e{2b(y_i)} \sum_{G \in \caT_n} \prod_{\{i,j\} \in G} |\zeta(y_i,y_j)| \biggr\} \\
&= \e{2b(x_1)} \exp\biggl\{ \int\dd|\mu|(y_1) \, |\zeta(x_1,y_1)| K_{N-1}(y_1) \biggr\}.
\end{split}
\ee
We have $K_{N-1}(y_1) \leq \e{a(y_1)+2b(y_1)}$ by the induction hypothesis. Eq.\ \eqref{bound K_N} follows from Assumption \ref{ass KP crit}.

The rest of the proof is standard combinatorics. The partition function can be expanded so as to recognize the exponential of connected graphs. Namely, we start with
\be
Z = 1 + \sum_{n\geq1} \frac1{n!} \int\dd\mu(x_1) \dots \int\dd\mu(x_n) \sum_{G \in \caG_n} \prod_{\{i,j\} \in G} \zeta(x_i,x_j).
\ee
The graph $G \in \caG_n$ can be decomposed into $k$ connected graphs whose sets of vertices form a partition of $\{1,\dots,n\}$. Summing first over the number $m_i$ of vertices for each set of the partition, we get
\be
\begin{split}
Z &= 1 + \sum_{n\geq1} \sum_{k\geq1} \frac1{k!} \sumtwo{m_1,\dots,m_k \geq 1}{m_1+\dots+m_k=n} \frac1{m_1! \dots m_k!} \\
&\hspace{35mm} \prod_{\ell=1}^k \Bigl\{ \int\dd\mu(x_1) \dots \int\dd\mu(x_{m_\ell}) \sum_{G \in \caC_{m_\ell}} \prod_{\{i,j\} \in G} \zeta(x_i,x_j) \Bigr\} \\
&= 1 + \sum_{n\geq1} \sum_{k\geq1} \frac1{k!} \sumtwo{m_1,\dots,m_k \geq 1}{m_1+\dots+m_k=n}  \prod_{\ell=1}^k \Bigl\{ \int\dd\mu(x_1) \dots \int\dd\mu(x_{m_\ell}) \varphi(x_1,\dots,x_{m_\ell}) \Bigr\}.
\end{split}
\ee
The triple sum is absolutely convergent thanks to the estimate \eqref{the bound} that we have just established. One can then interchange the sums by the dominated convergence theorem. This removes the sum over $n$, and this completes the proof of Theorem \ref{thm clexp}.
\end{proof}

Next we prove Theorems \ref{thm correlation functions} and \ref{thm decay correlation functions} in reverse order, since we will use the convergence properties in the latter theorem to get the former.

\begin{proof}[Proof of Theorem \ref{thm decay correlation functions}]
From the definition \eqref{def Z2 hat} and the tree estimate \ref{tree estimate}, we have
\be
|\hat Z(x_1,x_2)| \leq \sum_{n\geq2} \frac1{(n-2)!} \int\dd|\mu|(x_3) \dots \int\dd|\mu|(x_n) \prod_{i=1}^n \e{2b(x_i)} \sum_{G \in \caT_n} \prod_{\{i,j\}} |\zeta(x_i,x_j)|.
\ee
The expression above involves a sum over trees of arbitrary size that connect 1 and 2. Any such tree decomposes into a line of $m+1$ edges that connect 1 and 2 ($m\geq0$), and $m+2$ trees rooted in the vertices of the connecting line. Taking into account the combinatorial factors, we obtain
\bm
|\hat Z(x,y)| \leq |\zeta(x,y)| K(x) K(y) \\
+ \sum_{m\geq1} \int\dd|\mu|(x_1) \dots \int\dd|\mu|(x_m) \biggl( \prod_{i=0}^m |\zeta(x_i,x_{i+1})| K(x_i) \biggr) K(y) \nn
\end{multline}
with $x_0 \equiv x$ and $x_{m+1} \equiv y$. The result follows from the bound \eqref{bound K_N} for $K$.
\end{proof}

\begin{proof}[Proof of Theorem \ref{thm correlation functions}]
It is actually similar to the end of the proof of Theorem \ref{thm clexp}. $Z(x)$ can be expanded as a sum over graphs, that can be decomposed into a connected graph that contains 1, and other connected graphs. Taking into account the combinatorial factors, the contribution of connected graphs containing 1 yields $\hat Z(x)$, and the contribution of the others yields the expression \eqref{clexp} for $Z$. Thus $Z(x) = \hat Z(x) \, Z$. One step involved interchanging unbounded sums, which is justified because everything is absolutely convergent, thanks to \eqref{the bound} and Theorem \ref{thm decay correlation functions}.

In the graph expansion for $Z(x,y)$, the terms where 1 and 2 belong to the same connected graph yield $\hat Z(x,y) \, Z$, and the terms where 1 and 2 belong to different connected graphs yield $\hat Z(x) \hat Z(y) Z$. The detailed argument is the same as above. We then obtained the desired expression.
\end{proof}

\medskip
{\footnotesize
{\bf Acknowledgments:} We thank the referee for helpful comments. D.U. is grateful to Bill Faris, Roberto Fern\'andez, Roman Koteck\'y, Charles-\'Edouard Pfister, Alan Sokal, and Milo\v s Zahradn\'\i k for many useful discussions. He also acknowledges the hospitality of the Armenian Academy of Science, CNRS Marseille, the University of Geneva, ETH Z\"urich, the Center of Theoretical Studies of Prague, and the University of Arizona, where parts of this project were carried forward. D.U. is supported in part by the grant DMS-0601075 of the US National Science Foundation.
}

\end{document}